\documentclass{llncs}

\usepackage{geometry}
\usepackage{microtype}

\usepackage{graphicx}
\usepackage{algorithm}
\usepackage[noend]{algpseudocode}
\usepackage{xcolor}
\usepackage{hyperref}
\newenvironment{appendix-lemma}[1]{\noindent{\bf Lemma~#1~} \em }

\newcommand{\REM}[1]{}
\newcommand{\RR}{\mathcal{R}}
\newcommand{\HH}{\mathcal{H}}

\newcommand{\Master}{{\bf Master} }
\newcommand{\Shuffle}{{\bf Shuffle} }

\newcommand{\HRLQ}{{\sf HRLQ}}

\newcommand{\MinBP}{\mbox{{\sf Min-BP}}}
\newcommand{\MinBR}{\mbox{{\sf Min-BR}}}
\newcommand{\NPHard}{\mbox{{\sf NP-Hard}}}
\newcommand{\HR}{\mbox{{\sf HR}}}

\title{How good are Popular Matchings?}
\author{Krishnapriya A M\inst{1}\thanks{Part of this work was done when the author was an M.Tech. student at IIT Madras.}, 
        Meghana Nasre\inst{2}, 
        Prajakta Nimbhorkar\inst{3}, 
        Amit Rawat\inst{4} \thanks{Part of this work was done when the author was an MS student at IIT Madras.}}

\institute{Citrix Research and Development India, \and 
Indian Institute of Technology Madras, India, \and 
Chennai Mathematical Institute India and UMI ReLaX \and
University of Massachusetts Amherst, USA}

\begin{document}

\maketitle

\begin{abstract}
In this paper, we consider the Hospital Residents problem (\HR) and the Hospital Residents problem
with Lower Quotas (\HRLQ).
In this model with two sided preferences, stability is a well
accepted notion of optimality. However, in the presence of lower quotas, a stable and feasible matching
need not exist.
For the \HRLQ\ problem, our goal therefore is to output a {\it good} feasible matching
assuming that a feasible matching exists. 
Computing matchings with minimum number of blocking pairs (\MinBP) and minimum number of blocking residents (\MinBR) are known to be
NP-Complete. The only approximation algorithms for these problems work under severe restrictions on the preference
lists.  We present an algorithm which circumvents this restriction and computes a {\em popular} matching
in the \HRLQ\ instance. We show that 
on data-sets generated using various generators, our algorithm performs very well in terms of blocking 
pairs and blocking residents. Yokoi~\cite{Yokoi17} recently studied {\em envy-free} matchings for the \HRLQ\ problem. We
propose a simple modification to Yokoi's algorithm to output a {\em maximal envy-free} matching. 
 We observe that popular matchings outperform envy-free matchings on several parameters of practical importance, like size,
number of blocking pairs, number of blocking residents.

In the absence of lower quotas, that is, in the Hospital Residents (\HR) problem, stable matchings
are guaranteed to exist. Even in this case,  we show that popularity is a practical alternative
to stability.
For instance,
on synthetic data-sets generated using a particular model, as well as on real world data-sets,
  a popular matching is on an average 8-10\% larger in size, matches more number
of residents to their top-choice, and more residents prefer the popular matching as compared to a stable matching. 
Our comprehensive study reveals the practical appeal of popular matchings for the \HR\ and \HRLQ\ problems.
To the best of our knowledge, this is the first study on the empirical evaluation of popular matchings in this setting.
\end{abstract}


\section{Introduction}

In this paper, we study two problems -- the Hospital Residents (\HR) problem and the Hospital Residents problem with Lower Quotas (\HRLQ).
The input to the \HR\ problem is a bipartite graph $G = (\RR \cup \HH, E)$ where $\RR$ denotes a set of residents, $\HH$ denotes a
set of hospitals, and an edge $(r, h) \in E$ denotes that $r$ and $h$ are acceptable to each other. 
Each vertex has a {\em preference list} which is a strict ordering on its neighbors.
Further, each hospital has a positive upper-quota $q^+(h)$. In the \HRLQ\ problem, additionally, a hospital has a non-negative lower-quota $q^-(h)$.
A {\em matching} $M$ is a subset of $E$ such
that every resident is assigned at most one hospital and every hospital is assigned at most upper-quota many residents. 
Let $M(r)$ denote the hospital to which resident $r$ is matched in $M$. Analogously,
let  $M(h)$ denote the set of residents that are matched to $h$ in $M$. 
A matching $M$ in an \HRLQ\ instance is {\em feasible} if for every hospital 
$h$, $q^-(h) \le |M(h)| \le q^+(h)$.
The goal is to compute a feasible matching that is {\it optimal} with respect to the preferences of
the residents and the hospitals. 
When both sides of the bipartition express preferences, 
{\it stability} is a well-accepted notion of optimality. A stable matching is defined by the absence of a {\it blocking pair}.
\begin{definition}\label{def:stab}
A pair $(r, h) \in E \setminus M$ blocks $M$ if either $r$ is unmatched in $M$ or $r$ prefers $h$ over $M(r)$  
and either $|M(h)|<q^+(h)$ or $h$ prefers $r$ over some $r' \in M(h)$. A matching $M$ is {\em stable} if there
does not exist any blocking pair w.r.t. $M$, else $M$ is unstable.
\end{definition}

From the seminal result of Gale and Shapley~\cite{GS62}, it is known that every instance of the \HR\ problem admits a stable matching
and such a matching can be computed in linear time in the size of the instance. In contrast, there exist simple instances of the \HRLQ\ problem which do not admit
any matching that is both {\em feasible and stable}.

\begin{figure}[h]
\begin{minipage}{0.5\linewidth}
\begin{center}
\begin{tabular}{cccc}
$\bf{r_1:}$&$h_1$&$ h_2$\\
$\bf{r_2:}$&$h_1$&$ h_2$\\
$\bf{r_3:}$&$h_1$\\
\end{tabular}
\end{center}
\end{minipage}
\begin{minipage}{0.45\linewidth}
\begin{center}
\begin{tabular}{ccccc}
$\bf{[0,2] \hspace{0.1in} h_1:}$ &$r_1$& $r_2$&$r_3$\\
$\bf{[1,1] \hspace{0.1in} h_2:}$ &$r_2$&$r_1$\\
\end{tabular}
\end{center}
\vspace{0.1in}
\end{minipage}
\caption{A hospital residents instance $G$ with 
$\RR = \{r_1, r_2, r_3\}$ and 
$\HH = \{h_1, h_2\}$. The quotas of the hospitals are $q^-(h_1) = 0, q^+(h_1) = 2,  q^-(h_2)  = q^+(h_2) = 1$. 
The preferences can be read from the tables as follows: $r_1$ prefers $h_1$ over $h_2$ and so on. 
The three matchings $M_1$, $M_2$ and $M_3$ are shown below.  $M_2$ and $M_3$ are feasible but unstable since $(r_2,h_1)$ blocks $M_2$ and $(r_1,h_1)$ blocks $M_3$.
  $M_1 = \{(r_1, h_1), (r_2, h_1)\} \hspace{0.2in} M_2 = \{(r_1, h_1), (r_2, h_2), (r_3, h_1)\} \hspace{0.2in} M_3 = \{(r_1, h_2), (r_2, h_1), (r_3, h_1)\}$.}
\label{fig:example-intro}
\vspace{-0.1in}
\end{figure}

Figure~\ref{fig:example-intro} shows an \HRLQ\ instance where $h_2$ has a lower-quota of $1$. The instance  admits a unique
stable matching $M_1$ which is not feasible. The instance admits two maximum cardinality matchings $M_2$ and $M_3$, both of which
are feasible but unstable. 
This raises the question what is an optimal feasible matching in the \HRLQ\ setting?

Our goal in this paper is to propose {\em popularity} as a viable option in the \HRLQ\ setting, and compare and contrast it by an extensive experimental
evaluation with other approaches proposed for this problem. Before giving a formal definition of popularity, we describe some approaches from the literature to the \HRLQ\ problem.

Hamada~et~al.~\cite{Hamada16} proposed the optimality notions 
(i) minimum number of blocking pairs (\MinBP) or (ii) minimum number of residents that participate in any blocking pair (\MinBR). 
Unfortunately, computing a matching which is optimal according to any of the two notions is \NPHard\ as proved by Hamada~et~al.~\cite{Hamada16}. On the positive side, they give
approximation algorithms for a special case of \HRLQ, complemented by matching inapproximability results. Their
approximation algorithms require that all the hospitals with non-zero lower-quota have {\em complete}
preference lists ({\em CL-restriction}).

There are two recent works \cite{NN17,Yokoi17} which circumvent the CL-restriction and work with the natural assumption 
that the \HRLQ\ instance admits some feasible matching. Nasre and Nimbhorkar~\cite{NN17} consider the notion of {\em popularity} and show how to compute a feasible matching which
is a maximum cardinality popular matching in the set of feasible matchings. 
Yokoi~\cite{Yokoi17} considers the notion of {\em envy-freeness} in the \HRLQ\ setting.
We define both popularity and envy-freeness below.
     
\noindent {\bf Popular matchings: }Popular matchings are defined based on comparison of two matchings with respect to votes of the participants. 
To define popular matchings,
first we describe how residents and hospitals vote between two matchings $M$ and $M'$. We use
the definition from \cite{NN17}.
For a resident $r$ unmatched in $M$, we set $M(r)=\bot$ and assume that $r$ {prefers to be matched to any
hospital in his preference list (set of acceptable hospitals) over $\bot$}. Similarly, for a hospital $h$ with capacity $q^+(h)$,
we set the positions $q^+(h)-|M(h)|$ in $M(h)$ to $\bot$, so that $|M(h)|$ is always equal to $q^+(h)$.
For a vertex $u\in \RR\cup \HH$, $vote_u(x,y)=1$ if $u$ prefers $x$ over $y$, $vote_u(x,y)=-1$ if $u$ prefers $y$
over $x$ and $vote_u(x,y)=0$ if $x=y$. Thus, for a resident $r$, $vote_r(M,M')=vote_r(M(r),M'(r))$.

\noindent {\bf Voting for a hospital:} A hospital $h$ is assigned $q^+(h)$-many
votes to compare two matchings  $M$ and $M'$; this can be viewed as one vote
per position of the hospital. A hospital is indifferent between $M$ and $M'$ as far as its
$|M(h)\cap M'(h)|$ positions are concerned. For the remaining positions of hospital $h$,
the hospital defines a function ${\bf corr}_h$. This function allows $h$ to decide
any pairing of residents in $M(h)\setminus M'(h)$ to residents in
$M'(h)\setminus M(h)$. Under this pairing, for a resident $r \in M(h)\setminus M'(h)$, ${\bf corr}_h(r,M,M')$ is the resident in $M'(h)\setminus M(h)$ corresponding to $r$.
Then $vote_h(M,M')=\sum_{r\in M(h)\setminus M'(h)}vote_h(r,{\bf corr}_h(r,M,M'))$.
As an example, consider $q^+(h) = 3$ and $M(h) = \{r_1, r_2, r_3\}$ and $M'(h) = \{r_3, r_4, r_5\}$. 
To decide its votes, $h$ compares between $\{r_1, r_2\}$ and $\{r_4, r_5\}$ and one possible ${\bf corr}_h$ is to pair  $r_1$ with $r_4$ and $r_2$ with $r_5$.
Another ${\bf corr}_h$ function is to pair $r_1$ with $r_5$ and $r_2$ with $r_4$. The choice of the pairing of residents using ${\bf corr}_h$ is determined by the  hospital $h$.
With the voting scheme as above, popularity can be defined as follows:
\begin{definition} \label{def:pop}
A matching $M$ is more popular than  $M'$ if $\sum_{u\in \RR\cup\HH} vote_u(M,M')> \sum_{u\in \RR \cup \HH} vote_u(M',M)$. A matching $M$ is popular if there is no matching $M'$ more popular
than $M$.
\end{definition}
It is interesting to note that the algorithms for computing popular matchings do not need the ${\bf corr}$ function
as an input. The matching produced by the algorithms presented in this paper is popular for {\em any} ${\bf corr}$
function that is chosen.

\begin{definition} Given a feasible matching $M$ in
an \HRLQ\ instance, a resident $r$ has {\em justified envy} towards $r'$ with $M(r') = h$ if $h$ prefers $r$ over $r'$ and $r$ is either unmatched or prefers $h$ over $M(r)$. 
A matching
is envy-free if there is no resident that has a justified envy towards another resident. 
\end{definition}
Thus envy-freeness is a relaxation of stability. 
A matching that is popular amongst feasible matchings always exists \cite{NN17}, but
there exist simple instances of the \HRLQ\ problem that admit a feasible matching but do not admit a feasible envy-free matching. Yokoi~\cite{Yokoi17}
gave a characterization of \HRLQ\ instances that admit an envy-free matching
and an efficient
algorithm to compute some envy-free matching.
               
An envy-free matching
need not even be {\em maximal.} In the example in Figure~\ref{fig:example-intro} the matching $M = \{(r_2, h_2)\}$ is
envy-free.
Note that $M$ is not maximal, and not even a {\em maximal envy-free matching}. 
The matching $M' = \{(r_1, h_1), (r_2, h_2)\}$ is a {maximal envy-free}
matching, since addition of any  edge to $M'$ violates the envy-free property. The matchings $M_2$ and $M_3$ shown
in Figure~\ref{fig:example-intro} are not envy-free,
since $r_2$ has
justified envy towards $r_3$ in $M_2$ whereas $r_1$ has a justified envy towards $r_3$ in $M_3$. 
The algorithm in ~\cite{Yokoi17} outputs the matching $M$ which need not even be maximal; Algorithm~\ref{algo:envy-free} in Section~\ref{sec:envy} outputs $M'$  which is guaranteed to be maximal envy-free. Finally,  the
Algorithm~\ref{algo:hrlq} in Section~\ref{sec:pop} outputs $M_2$
(see Figure~\ref{fig:example-intro}) which is both maximum cardinality as well as popular. 

Popularity is an interesting alternative to stability even in the absence of lower-quotas, that is, in the \HR\ problem.
The \HR\ problem is motivated by large scale real-world applications like the National Residency Matching Program (NRMP)
in US~\cite{nrmp} and the Scottish Foundation Allocation Scheme (SFAS) in Europe~\cite{sfas}.
In applications like NRMP and SFAS it is desirable from the social perspective to match 
as many residents as possible so as to reduce 
the number of unemployed residents and to provide better staffing to hospitals.
Bir{\`{o}}~et~al.~\cite{BiroMM10} report that for the SFAS data-set (2006--2007) the administrators were interested in knowing
if relaxing stability
would lead to gains in the size of the matching output.
It is known that the size of a stable matching could be half that of the maximum cardinality matching $M^*$,
whereas the size of a maximum cardinality popular matching is at least $\frac{2}{3}|M^*|$ (see e.g. \cite{Kavitha14}). Thus, theoretically, popular matchings give an attractive alternative to stable matchings. 

Popular matchings have gained lot of attention and there have been
several interesting results in the recent past ~\cite{AIKM07,BIM10,Cseh17,BrandlK16,Kavitha14,NN17,NR17}.
The algorithms used and developed in this paper, are inspired by a series
of papers \cite{BrandlK16,CsehK16,Kavitha14,NN17,NR17}.
One of the goals in this paper is
to complement these theoretical results with an extensive experimental evaluation of the quality
of popular matchings.
We now summarize our contributions below:

\subsection { Our contribution}
\noindent{\bf Results for \HRLQ:}
{\bf Algorithm for popular matchings:} We propose a variant of the algorithm in \cite{NN17}. 
Whenever the input \HRLQ\ instance admits a feasible matching, our algorithm outputs a (feasible) popular matching amongst
all feasible matchings. We report the following experimental findings:
\begin{itemize}
\item {\em \MinBP\ and \MinBR\ objectives: }
In all the data-sets, the matching output by our algorithm is { at most $3$ times the optimal} for 
the \MinBP\ problem and { very close to optimal} for the \MinBR\ problem.
In the \HRLQ\ setting, no previous approximation algorithm is known for incomplete preference lists. 
The only known algorithms which output matchings with theoretical guarantees for the \MinBP\ and \MinBR\ problems work under the CL-restricted model~\cite{Hamada16}.
We remark that
the algorithm of \cite{NN17} can not provide any bounded approximation ratio for \MinBP\ and \MinBR\ 
problems as it may output a feasible matching with non-zero blocking pairs and non-zero
blocking residents even when the instance admits a stable feasible matching. 

\item {\em Comparison with envy-free matchings: }Based on the algorithm by Yokoi~\cite{Yokoi17}, we give an algorithm to compute a maximal 
envy-free matching, if it exists, and 
compare its {\em quality} with a popular matching. 
Besides the fact that popular matchings exist whenever feasible matchings exist, which is not the case with envy-free matchings, our experiments
illustrate that, compared to an envy-free matching, a popular matching is about {32\%--43\%} larger, matches about {15\%--38\%}
more residents to their top choice and has { an order of magnitude fewer} blocking pairs and blocking residents.
\end{itemize}

\noindent{\bf Empirical Evaluation of known algorithms for  \HR\ :} For the \HR\ problem, we implement and extensively evaluate known algorithms for maximum cardinality popular matching and popular
amongst maximum cardinality matching
on
synthetic data-sets as well as limited number of real-world data-sets.
\begin{itemize}
\item Our experiments show that 
a maximum cardinality popular matching, as well as popular amongst maximum cardinality matchings are 8--10$\%$ larger in size than a stable matching. 
We also observe that a maximum cardinality popular matching almost always fares better than a stable matching with respect to the
 number of residents 
matched to their rank-1 hospitals, and the number of residents favoring a maximum cardinality popular matching over a stable matching. 
\end{itemize}
We note that these properties cannot be proven theoretically as there are instances where they do not hold.  
Despite these counter-examples (see Appendix~\ref{sec:app_example})
our empirical results show that the desirable properties hold on synthetic as well as real-world instances. 
Hence we believe that popular matchings are a very good practical alternative.

\noindent {\bf Organization of the paper:} 
In Section~\ref{sec:pop}, we present our algorithm
to compute a matching that is popular amongst feasible matchings for the \HRLQ\ problem. 
Our algorithm for computing a maximal envy-free matching is described in Section \ref{sec:envy}.
In Section~\ref{sec:setup} we give details of our experimental setup and the data-generation models developed by us. Section~\ref{sec4} gives our empirical results for the \HRLQ\ and \HR\ problems. 
We refer the reader to \cite{BrandlK16,NR17} for known algorithms for computing popular matchings in the \HR\ problem. 


\section{Algorithms for \HRLQ\ problem}
In Section \ref{sec:pop}, we 
present our algorithm to compute a popular matching in the \HRLQ\ problem. As mentioned earlier, our algorithm
is a variant of the algorithms in ~\cite{NN17}. The algorithm to find a maximal envy-free matching is given in Section \ref{sec:envy}, and its
output is a superset of the envy-free matching
computed by Yokoi's algorithm \cite{Yokoi17}.

\subsection{Algorithm for Popular matching in \HRLQ}\label{sec:pop}
In this section, we present an algorithm (Algorithm~\ref{algo:hrlq}) that outputs a feasible matching $M$ in a given \HRLQ\ instance $G$, 
such that $M$ is popular amongst all the feasible matchings in $G$.
We assume that $G$ admits a feasible matching, which can
be checked in polynomial time~\cite{Hamada16}.

Algorithm~\ref{algo:hrlq} is a modification of
the standard hospital-proposing Gale and Shapley algorithm~\cite{GS62} and can be viewed as a 
two-phase algorithm.
The first phase simply computes a stable matching $M_s$ in the input instance ignoring lower quotas.
If $M_s$ satisfies the lower quotas of all the hospitals, it just outputs $M_s$. Otherwise
hospitals that are {\it deficient} in $M_s$ propose with {\it increased priority}. A hospital $h$ is deficient w.r.t. 
a matching $M$
if $|M(h)| < q^-(h)$. 
We implement the increased priority by assigning a {\em level}
to each hospital, so that higher level corresponds to higher priority. A resident always prefers a higher level hospital in its preference list to a lower level hospital, 
irrespective of their relative positions in its preference list.
We show that at most $|\RR|$ levels suffice to output a feasible matching in the instance, if one exists. We give a detailed description below.

\begin{algorithm}[ht]
\footnotesize
\begin{algorithmic}[1]
\State Input : $G = (\RR \cup \HH, E)$
\State set $M = \emptyset$; $Q = \emptyset$; \label{init}
\For {each $h \in \HH$} \label{hos-levels}
	set level($h$) = 0; add $h$ to $Q$;
\EndFor
\For {each $r \in \RR$} \label{res-levels}
	set level($r$) = -1;
\EndFor

\While {$Q$ is not empty} \label{while}
\State let $h =$ head($Q$);\label{GS-beg}
\State let $\mathcal{S}_h$ = residents to whom $h$ has not yet proposed with level($h$);
\If {$\mathcal{S}_h \neq \emptyset$ } \label{Sh-check}
        \State let $r$ be the most preferred resident in $\mathcal{S}_h$ \label{propose}
        \If {$r$ is matched in $M$ (to say $h'$)} \label{check}
		\If { {\bf pref}($r, h, h'$) == 0}

			\State add $h$ to $Q$;
			 goto Step~\ref{while};
        	\Else

                	\State $M = M \setminus \{(r, h')\}$; \label{remove}
			\State add $h'$ to $Q$ if $h'$ is not present in $Q$; \label{add-h-prime}
		\EndIf 
	\EndIf
        \State $M = M \cup \{(r, h)\}$;  level($r$) = level($h$); \label{matchr}
        \If { level($h$) == 0 and $|M(h)|< q^+(h)$} \label{all-h}

                \State add $h$ to $Q$;
        \ElsIf { $h \in \HH_{lq}$ and $|M(h)| < q^-(h)$} \label{def-h}
                \State add $h$ to $Q$;
        \EndIf
\ElsIf  { $h \in \HH_{lq}$ and level($h$) < $|\RR|$} \label{inc-level}
	\State level($h$) = level($h$)+1; add $h$ to $Q$;
	\State // $h$ starts proposing from the beginning of the preference list.
\EndIf
\EndWhile
\end{algorithmic}
\caption{Popular matching in \HRLQ}
\label{algo:hrlq}
\end{algorithm}

\noindent Let $G = (\RR \cup \HH, E)$ be the given \HRLQ\ instance. 
 Let $\HH_{lq}$ denote the set of hospitals
which have non-zero lower quota -- we call $h \in \HH_{lq}$ a lower-quota hospital; we call $h \notin \HH_{lq}$ a non-lower quota hospital.

Algorithm~\ref{algo:hrlq} begins by initializing the matching $M$ and a queue $Q$ of hospitals to be empty (Step~\ref{init}). As described 
above, level of a hospital denotes its current priority, and initially all the hospitals have
level $0$. All hospitals are added to $Q$ (Step~\ref{hos-levels}).
We also assign a level to each resident $r$, which stores the level of the hospital 
$M(r)$ at the time when $r$ gets matched to $M(r)$. Initially, all residents are assigned level $-1$ (Step~\ref{res-levels}). 
The main while loop of the algorithm (Step~\ref{while}) executes as long as $Q$ is non-empty. 
Steps~\ref{GS-beg}--\ref{matchr} are similar to  the execution of the hospital-proposing Gale-Shapley algorithm~\cite{GI89}.
When a matched resident $r$ gets a proposal from a hospital $h$, $r$ compares its current match $h'$ with  $h$ using {\bf pref}$(r,h,h')$.
We define 
{\bf pref} $(r,h,h')=1$ if 
$(i)$ level$(r) < $ level$(h)$, which means $h$ proposes to $r$ with a higher priority than $h'$ did, or $(ii)$ level$(r) = $ level$(h)$ and $h$ has a higher position than $h'$
in the preference list of $r$. 
We define 
{\bf pref} $(r,h,h')=0$ otherwise.
Thus matched resident $r$ accepts the proposal of $h$ and rejects $h'$ if and only if {\bf pref} $(r,h,h')=1$.
In case the edge $(r, h)$ is added to $M$ (Step~\ref{matchr}), the algorithm also sets the level of $r$ equal to the level($h$). 

When a hospital $h$ finishes proposing all the residents on its preference list with level$(h)=i$ and
still $|M(h)|<q^-(h)$, level$(h)$ is incremented by $1$ and $h$ is added back to $Q$.
 This is done in Step~\ref{inc-level}.
Now $h$ restarts proposing all the residents in its preference list with the new value of
level$(h)$.

\noindent {\bf Running time and correctness:} To see that the algorithm terminates we observe that every non-lower quota hospital
proposes to residents in its preference list at most once. Every lower-quota hospital proposes to residents
on its preference list at most $|\RR|$ times. Thus the running time of our algorithm is $O((|\RR|+|\HH| + |E|) \cdot |\RR|)$.
To see the correctness note that  when $G$ admits a feasible stable matching  our algorithm degenerates to the standard Gale-Shapley hospital
proposing algorithm. As proved in \cite{NR17},
a stable matching is popular amongst the set of feasible matchings.
In case $G$ admits a feasible matching but no feasible stable matching,   techniques as in \cite{NN17} can be employed to show that the output
is popular amongst the set of feasible matchings.
\begin{theorem}
In an \HRLQ\ instance $G$, Algorithm~\ref{algo:hrlq} outputs a matching that is feasible and popular amongst the set of feasible matchings.
If $G$ admits a feasible and stable matching, then Algorithm~\ref{algo:hrlq} outputs a stable matching.
\vspace{-0.1in}
\end{theorem}

\subsection{Algorithm for Maximal Envy-free Matching}\label{sec:envy}
Yokoi \cite{Yokoi17} has given an algorithm to compute an envy-free matching in an \HRLQ\ instance $G=(\RR\cup\HH,E)$.
Yokoi's algorithm works in the following steps: 
\begin{enumerate}
\item Set the upper quota of each hospital to its lower quota.
\item Set the lower quota of each hospital to $0$, call this modified instance $G_1$.
\item Find a stable matching $M_1$ in the modified instance.
\item Return $M_1$, if every hospital gets matched to as many residents as its upper quota in $G_1$, otherwise declare that there is no envy-free matching in $G$.
\end{enumerate}

It can be seen that Yokoi's algorithm, as stated above, does not return a {\em maximal envy-free} matching.
In particular, any hospital with lower-quota $0$ does not get matched to any resident in Yokoi's algorithm. Hence we
propose a simple extension of Yokoi's algorithm to compute a maximal envy-free matching, which contains the matching output by Yokoi's algorithm.
 We call a matching $M$ in $G$ maximal envy-free if, for any 
edge $e=(r,h)\in E\setminus M$, $M\cup  \{e\}$ is not envy-free. The following definition with respect to Yokoi's output matching $M_1$ 
 is useful for our algorithm.
 Our algorithm is described
as Algorithm~\ref{algo:envy-free}.
\begin{definition}
Let $h$ be a hospital that is under-subscribed in $G$ with respect to $M_1$, that is $|M_1(h)| < q^+(h)$. 
A {\em threshold resident} $r_h$ for $h$, if one exists, is the most preferred  resident  in the preference list of $h$ 
such that
$(i)$ $r_h$ is matched in $M_1$, to say $h'$, and $(ii)$ $r_h$ prefers $h$ over $h'$. If no such resident exists,
we assume a unique dummy resident $r_h$ at the end of $h$'s preference list to be the threshold resident for $h$.
\end{definition}
\begin{algorithm}[ht]
\footnotesize
\begin{algorithmic}[1]
\State Input : $G = (\RR \cup \HH, E)$
\State Compute a matching $M_1$ by Yokoi's algorithm.
\If{Yokoi's algorithm declares ``no envy-free matching"}
\State Return $M=\emptyset$.
\EndIf
\State Let $\RR'$ be the set of residents unmatched in $M_1$.
\State Let $\HH'$ be the set of hospitals such that $|M_1(h)|<q^+(h)$ in $G$.
\State Let $G'=(\RR'\cup\HH',E')$ be an induced subgraph of $G$, where $E' =
\{(r,h)\mid r\in\RR', h\in\HH',  \textrm{$h$ prefers $r$ over its  threshold resident $r_h$} \}$. Set $q^+(h)$ in $G'$ as $q^+(h) - q^-(h)$ in $G$.

\State Each $h$ has the same relative ordering on its neighbors in $G'$ as in $G$. 
\State\label{step:GS}$M_2$ =  stable matching in $G' = (\RR' \cup \HH', E')$.  
\State Return $M=M_1\cup M_2$.
\end{algorithmic}
\caption{Maximal envy-free matching in \HRLQ}\label{algo:envy-free}
\end{algorithm}

\noindent Below we prove that the output of Algorithm~\ref{algo:envy-free} is a maximal envy-free matching.
\begin{theorem}\label{thm:envy}
If $G$ admits an envy-free matching, then Algorithm~\ref{algo:envy-free} outputs $M$ which is maximal envy-free in $G$.
\end{theorem}
\begin{proof}
Since $G$ admits an envy-free matching, $M_1$ output by Yokoi's algorithm is non-empty.
We prove that $M$ is an envy-free feasible matching, and for each edge $e \in E\setminus M$, $M\cup\{e\}$ is not an envy-free matching.

\noindent {\bf $M$ is envy-free:}
Assume, for the sake of contradiction, that a resident $r'$ has a justified envy towards a resident $r$ with respect to $M$. 
Thus $r'$ prefers $h=M(r)$ over $h'=M(r')$, and $h$ prefers $r'$ over $r$. The edge $(r, h)\in M$ and
hence either $(r,h)\in M_1$ or $(r,h)\in M_2$. 
Suppose $(r, h) \in M_1$. Recall that $M_1$ is stable in the instance $G_1$ used in Yokoi's algorithm.
In this case, the edge $(r', h)$ blocks $M_1$ in $G_1$
a contradiction to the stability of $M_1$ in $G_1$. 

Thus, if possible, let  $(r, h)\in M_2$, and hence $(r,h)\in E'$. If $r'$ is unmatched in $M_1$, then $(r',h)\in E'$ and 
$(r', h)$ blocks $M_2$ in $G'$, contradicting the stability of $M_2$ in $G'$.
If $r'$ is matched in $M_1$ then $r'\notin\RR'$ and hence $(r',h)\notin E'$. In this case, the threshold resident $r_h$ of $h$ is either same as $r'$ or
is a resident whom $h$ prefers over $r'$. Since $h$ prefers $r'$ over $r$, $h$ prefers $r_h$ over $r$. Therefore $(r,h)$ can not be in $E'$ by construction, and 
hence $(r,h)\notin M_2$.
This proves that $M$ is envy-free. 

\noindent {\bf $M$ is maximal envy-free:} We now prove that, for any $e=(r, h)\notin M$, $M\cup\{e\}$ is not envy-free.
Let $e=(r, h), h\in\HH,r\in\RR$. Clearly, $h$ must be {\em under-subscribed in $M$} i.e. $|M(h)|<q^+(h)$, and $r$ must be unmatched in $M$, 
otherwise $M\cup\{e\}$ is not a valid matching in $G$.
Let $r_h$ be the threshold resident of $h$ with respect to $M_1$. Since $r$ is unmatched in $M$ and hence in $M_1$, $r\in \RR'$.
 (i) If $h$ prefers $r$ over $r_h$, then 
$(r, h)\in E'$ blocks $M_2$ in $G'$,
 which contradicts  the stability of  $M_2$ in $G'$. 
(ii) If $h$ prefers $r_h$ over $r$, then adding the edge $(r, h)$ to $M$
makes $r_h$ have a justified envy towards $r$. Thus,  $M\cup\{(r, h)\}$ is not envy-free. 
\end{proof}


\section{Experimental Setup}
\vspace{-0.1in}
\label{sec:setup}
The experiments were performed on a machine with a single
Intel Core i7-4770 CPU running at 3.40GHz, with 32GB of DDR3 RAM at 1600MHz.
The OS was Linux (Kubuntu 16.04.2, 64 bit) running kernel 4.4.0-59.
The {code}~\cite{code} 
is developed in C++ and was compiled using the clang-3.8 compiler with -O3 optimization level.
To evaluate the performance of our algorithms, we developed data generators
which model preferences of the participants in real-world instances.
We use a limited number of publicly available data-sets as well
as real-world data-sets from elective allocation at IIT-Madras.
All the data-sets generated and used by us are available at \cite{data}. 
\vspace{-0.1in}
\subsection{Data generation models and available data-sets}

There are a variety of parameters and all of them could be varied to generate synthetic data-sets.
We focus on four prominent parameters -- 
$(i)$ the number of residents $|\RR|$, 
$(ii)$ the number of hospitals $|\HH|$,
$(iii)$ the length of the preference list of each resident $k$, 
$(iv)$ capacity of every hospital $cap$ (by default $cap = |\RR| / |\HH|$).
In the synthetic data-sets, all  hospitals have uniform capacity and all residents have the same length of preference list.
We use the following three models of data generation,
and data-sets publicly available from~\cite{public}. 
\begin{enumerate}
\item {\bf Master}:
\label{subsec:mahdian_with_master}
Here, we model the real-world scenario that there are some hospitals which are in high demand
among residents and hence have a larger chance of appearing in the preference
list of a resident. 
Hospitals on the other hand, rely on some global criteria to rank residents. 
We set up a geometric probability distribution with $p=0.10$
over the hospitals which denotes the probability with which a hospital is
chosen for being included in the preference list of a resident.
Each resident samples $k$ hospitals according to the distribution and orders them arbitrarily.
We also assume that there exists a master list over the set of residents.
For all the neighbours of a hospital, the hospital ranks them according to 
the master list of residents.
\item {\bf Shuffle}:
\label{subsec:mahdian_with_shuffle}
This model is similar to the first model except that we do not assume a master
list on the residents. The residents draw their preference lists as above.
Every hospital orders its neighbours uniformly at random. 
This models the scenario that hospitals may have custom defined ranking criteria.
Our \Shuffle model is closely inspired by the one described by Mahdian and Immorlica~\cite{IM05}.

\item {\bf Publicly available HR with couples data-set}:
\label{subsec:hr_couple}
Other than generating data-sets using the three models described above, we used
a freely available data-set~\cite{public} by Manlove~et~al.~\cite{ManloveMT2017}
for the HR problem with couples (HRC problem).
The instances were only modified with respect to the preference list of
residents that participate in a couple, all other aspects of the instance
remain the same.
Residents participating in a couple can have different copies of a same
hospitals on their preference list which are not necessarily contiguous.
The preference list is created by only keeping the first unique copy of a hospital in the
preference list of a resident.

\item {\bf Elective allocation data from IIT-M}:
\label{subsec:iitmdata}
We use a limited number of real-world data-sets available from the IIT Madras elective
allocation. The data-sets are obtained from the SEAT (Student Elective Allocation Tool) which
allocates humanities electives and outside department electives for under-graduate students across  the institute every semester. The input consists of
a set of students and a set of courses with capacities (upper-quotas). Every student submits a preference ordering
over a subset of courses and every course ranks students based on custom ranking criteria.
This is exactly the HR problem.
\end{enumerate}

\REM{
\noindent {\bf Relation to other models}: 
Our \Shuffle model is closely inspired by the one described by Mahdian and Immorlica~\cite{IM05}.
There they work with a stable marriage setting and they assume an arbitrary
distribution on women (instead of the geometric distribution on hospitals in our case).
Manlove~et~al.~\cite{IrvingM2010} consider the problem of computing {\it large weakly stable} matchings
in an HR instance with ties allowed in the preference lists.
In their data-sets~\cite{IrvingM2010}, as well as synthetic data-sets generated by us, the number of
residents is fixed to be equal to the number of {\it positions} (equal to $n_1 \cdot cap$)
and the resident preference list of length is fixed to a constant value $5$.
Manlove~et~al.~\cite{ManloveMT2017} studied the problem of
finding a matching with the minimum number of blocking pairs in the case of
HR problem with couples.
The model used by them is similar to the \Shuffle model used
by us and we also use some of their publicly available data-sets.

}
\noindent {\bf \HRLQ\ instances:} The above three data-generation models (Master, Shuffle, and Random) are common for generating preferences in \HRLQ\ and \HR\ data-sets.
For \HRLQ\ data-sets we  additionally need lower quotas.
In all our data-sets, around 90\% of the hospitals had lower quota at least 1.
This is to ensure that the instances are \HRLQ\ instances rather than {\it nearly} HR instances.
Also, the sum of the lower quota of all the hospitals was kept around 50\% of
the total number of positions available.
This ensures that at least half of the residents must be matched to a
lower quota hospital.
Lastly, we consider only those \HRLQ\
instances which admit a feasible matching but {\em no stable feasible matching}.
This is done by simply discarding instances that admit a feasible stable matching.
We discard such instances because if an instance with feasible stable matching, the stable matching is optimal with respect to popularity, envy-freeness and 
min-BP and min-BR objectives.

\noindent{\bf Methodology}: For reporting our results, we fix a model of generation and
the parameters $|\RR|, |\HH|, k, cap$. For the chosen model and the parameters, we generate 10
data-sets and report arithmetic average for all output parameters on these data-sets.

\vspace{-0.1in}

\section{Empirical Evaluation}
\label{sec4}
Here, we present  our empirical observations on \HRLQ\ and \HR\ instances. In each case, we define
set of parameters on the basis of which we evaluate the quality of our matchings.
\vspace{-0.1in}
\subsection { \HRLQ\ instances}
In this section we show the performance of Algorithm~\ref{algo:hrlq} and Algorithm~\ref{algo:envy-free}
on data-sets generated using models described earlier. 
Let $G$ be an instance of the \HRLQ\ problem, and $M_s$ be the stable matching in the instance ignoring lower-quotas. For all our instances
$M_s$ is infeasible. Let $M_p$ denote  a popular matching output 
Algorithm~\ref{algo:hrlq} in $G$. Let $M_e$ denote a maximal envy-free matching output by Algorithm~\ref{algo:envy-free}.
Both $M_p$ and $M_e$ are feasible for $G$ and hence unstable.

\noindent {\bf \underline{Parameters of interest:}}
 We now define the parameters of the matching that are of interest in the \HRLQ\ problem.
For $M \in \{M_p, M_e\}$ we compute the following. 
\begin{itemize}
  \item $S(M):$ size of the matching $M$ in $G$.
  \item $BPC(M):$ number of blocking pairs w.r.t. $M$ in $G$.
    Since $M$ is not stable in $G$, this parameter is expected to be positive; however we
would like this parameter to be small. 
  \item $BR(M):$ number of residents that participate in at least one blocking
    pair w.r.t. $M$.
  \item $\RR_1(M):$ number of residents matched to their rank-1 hospitals in $M$.
\end{itemize}
We additionally compute the deficiency of the stable matching $M_s$.
Hamada~et~al.~\cite{Hamada16} showed that $Def(M_s, G)$ is a lower bound on the number of blocking pairs and the
number of blocking residents in an \HRLQ\ instance. 
To analyze the goodness of $M_p$ and $M_e$ w.r.t. the \MinBP\ and \MinBR\ objectives, we compare the number of
blocking pairs and blocking residents with the lower bound of $Def(M_s, G)$. 
\begin{itemize}
  \item $Def(M_s, G):$ This parameter denotes the deficiency of the stable (but not feasible) matching $M_s$ in $G$.  For every hospital $h$,
let  $def(M_s, h) =\max\{ 0,  q^-(h) - |M_s(h)|\}$. The deficiency of the instance $G$ is the sum of  the  deficiencies of all hospitals.
\end{itemize}

\noindent We now describe our observations for the data-sets generated using
various models.
For a particular model, we vary the parameters $|\RR|, |\HH|, k, cap$ to
generate \HRLQ\ data-sets using the different models.
In all our tables a column with the legend $\uparrow$ implies that larger values are better. Analogously, a column with the legend $\downarrow$ implies
smaller values are better.
\subsubsection{Popular Matchings versus Maximal Envy-free Matchings}
Here we report the quality of popular matchings and envy-free matchings on the parameters of interest listed above
on different models.

\begin{table}[h]

\begin{tabular}{|c||c||c|c||c|c||c|c||c|c|}
\hline
$|\HH|$& $Def(M_s, G)$& \multicolumn{2}{|c|}{$S(M)$ {\bf $\uparrow$}} & \multicolumn{2}{|c|}{$BPC(M)$ $\downarrow$} & \multicolumn{2}{|c|} {$BR(M)$ $\downarrow$} & \multicolumn{2}{|c|}{$\RR_1(M)$ $\uparrow$} \\
\hline
\hline
 &  & $M_p$ & $M_e$ & $M_p$ & $M_e$ & $M_p$ & $M_e$ & $M_p$ & $M_e$ \\

\hline
100 & 30.80 & {\bf 885.40} &559.00 & {\bf 78.50} & 2747.00 & {\bf 34.80} & 822.00 & {\bf 554.10 }&174.00  \\
20 & 23.60 & {\bf 897.90} & 510.66 &  {\bf 67.70} &3067.33 & {\bf  27.50} & 803.66  &{\bf  570.40} & 195.33 \\
10 & 27.50   &{\bf  912.80} & 535.50 & {\bf 85.10} &2945.00 &{\bf  31.10} & 770.87 &{\bf  600.40} & 226.87 \\
\hline
\end{tabular}
\vspace{0.1in}
\caption{Data generated using the \Master model. All values are absolute.
$|\RR| = 1000, k = 5$.} 
\label{tab:HRLQ-master-1000}
\end{table}
Table~\ref{tab:HRLQ-master-1000} shows the results for  popular matchings ($M_p$) and maximal envy-free matchings ($M_e$) on data-sets generated using the \Master model.

Table~\ref{tab:HRLQ-shuffle-1000} shows the results for  popular matchings ($M_p$) and maximal envy-free matchings ($M_e$) on data-sets generated using the \Shuffle model.

\begin{table}[h]

\begin{tabular}{|c||c||c|c||c|c||c|c||c|c|}
\hline
$|\HH|$& $Def(M_s, G)$& \multicolumn{2}{|c|}{$S(M)$ {\bf $\uparrow$}} & \multicolumn{2}{|c|}{$BPC(M)$ $\downarrow$} & \multicolumn{2}{|c|} {$BR(M)$ $\downarrow$} & \multicolumn{2}{|c|}{$\RR_1(M)$ $\uparrow$} \\
\hline
\hline
 &  & $M_p$ & $M_e$ & $M_p$ & $M_e$ & $M_p$ & $M_e$ & $M_p$ & $M_e$ \\

\hline
100 & 17.00 & {\bf 892.70} & -- & {\bf 27.20} & -- & {\bf 19.40} & --  & {\bf 350.30 }&  --\\
20 & 20.60 & {\bf 915.30} & 547.00  &  {\bf 34.40} & 2838.20 & {\bf  23.60} & 808.00  &{\bf  343.90} & 185.80 \\
10 & 35.40   &{\bf  930.00} & 490.33  & {\bf 57.80} & 3388.00 &{\bf  35.40} & 853.00  &{\bf  309.30} &  147.00 \\
\hline

\end{tabular}

\vspace{0.1in}
\caption{Data generated using the \Shuffle model. All values are absolute.
$|\RR| = 1000, k = 5$.} 
\label{tab:HRLQ-shuffle-1000}
\vspace{-0.3in}
\end{table}

We observe the following from the above two tables.
\begin{itemize}
\item {\bf Guaranteed existence:} As noted earlier, envy-free matchings are not guaranteed to exist in contrast to popular matchings which
always exist in \HRLQ\ instances.
For instance, in the \Shuffle model, for  $|\RR| = 1000, |\HH| = 100, k=5$ (Table~\ref{tab:HRLQ-shuffle-1000}, Row~1) none of the instances 
admit an envy-free matching.
Thus, for the columns $M_e$ we take an average over the instances that admit an envy-free matching.
\item {\bf Size:} It is evident from the tables that in terms of 
size popular matchings are about {32\%--43\%}  larger as compared to envy-free matchings 
(when they exist). See Column $S(M)$ in Table~\ref{tab:HRLQ-master-1000} and Table~\ref{tab:HRLQ-shuffle-1000}.
\item {\bf BPC and BR:} In terms of the blocking pairs and blocking residents, popular matchings beat envy-free matchings
by over an order of magnitude. We remark that to ensure envy-freeness, several hospitals may to be left under-subscribed. This
explains the unusually large blocking pairs and blocking residents in envy-free matchings. Furthermore, note that $Def(M_s, G)$ is a lower-bound on both the number of blocking pairs
and blocking residents. On all instances, for popular matchings the number of blocking pairs (BPC) is {at most $3$ times the optimal}
whereas the number of blocking residents (BR) is {close to the optimal value}.
\item {\bf Number of Envy-pairs:}  Although we do not report it explicitly, the number of envy pairs in an envy-free matching is trivially
zero and for any matching it is upper bounded by the number of blocking pairs. Since the number of blocking pairs is significantly small
for popular matchings, we conclude that the number of envy-pairs is also small.

\item {\bf Number of residents matched to rank-1 hospitals:} 
For any matching, a desirable criteria is to match as many participants to their top-choice. Again on this count,
we see that popular matchings match about {15\%--38\% } more residents to their top choice hospital (See Column $\RR_1(M)$ in the above tables).

\end{itemize}


\subsection{Results on HR instances}
\label{sec:res-HR}
To analyze the quality of popular matchings on various data-sets,
we generate an HR instance $G = (\HH \cup \RR, E)$,  compute a resident optimal stable matching $M_s$, a maximum cardinality
popular matching $M_p$ and a popular matching among maximum cardinality matchings $M_m$.
The matchings $M_p$ and $M_m$ are generated by creating  reduced instances $G_2$ and $G_{|\RR|}$ respectively,
and executing the resident proposing Gale-Shapley algorithm in the respective instance.
Theoretically, we need to construct $G_{|\RR|}$ for computing $M_m$, practically
we observe that almost always a constant number suffices say constructing $G_{10}$ suffices.
We now describe our output parameters. 

\noindent\underline{{\bf Parameters of interest:}}
For any matching $M$, let $\RR_1(M)$ denote the number of residents matched to rank-1 hospitals in $M$.
For two matchings $M$ and $M'$, let  $\mathcal{V}_{\RR}(M, M')$ denote the number of residents that prefer $M$ over $M'$.
For each instance we compute the following:
\begin{itemize}
\item $S(M_s):$  size of the stable matching $M_s$ in $G$.
\item For $M$ to be one of $M_p$ or $M_m$ define:
\begin{itemize}
\item $\Delta:$ $\frac{|M| - |M_s|}{|M_s|} \times 100$. This denotes the percentage increase in size of $M_s$ when compared to $M$.
When comparing $M_s$ with either $M_p$ or $M_m$, $\Delta$ is guaranteed to be non-negative.
    The larger this value, the better the matching in terms of size as compared to $M_s$. 
\item $\Delta_1 :$  $\frac{\RR_1(M) - \RR_1(M_s)}{\RR_1(M_s)} \times 100$.
  This denotes the percentage increase in number of rank-1 residents
  of $M_s$ when compared to $M$.
  As discussed in the Introduction, there is no guarantee that this value is non-negative.
  However, we prefer that the value is as large as possible.
\item $\Delta_{\RR} :$  $\frac{\mathcal{V}_{\RR}(M, M_s)- \mathcal{V}_{\RR}(M_s, M)}{ |\RR| }\times 100$.
  This denotes the percentage increase in number  of resident votes of $M$ when compared to $M_s$.
  Similar to $\Delta_1$, there is no apriori guarantee that more residents prefer $M$ over $M_s$.
  A positive value for this parameter indicates that $M$ is {\it more resident popular}
  as compared to  $M_s$. That is, in an election where only residents
  vote, a majority of the residents would like to move from $M_s$ to $M$. 
\REM{  It is well known that \cite{GI89} there does not exist any matching $M'$ stable or otherwise
  such that a subset of  residents are better off in $M'$ as compared to $M_s$ and 
  no resident is worse off in $M'$ as compared to $M_s$.
  However, a positive value for $\Delta_{\RR}$ in our results for $M$ (either $M_p$ or $M_m$)
  does not contradict the above.
  This is because in the matching $M$  some residents are allowed to be worse off in $M$ as compared to $M_s$.}
\item $BP(M):$ $\frac{\mbox{number of blocking pairs in $M$}} {|E| - |M|} \times 100$.
  The minimum value for $BP(M)$ can be $0$ (for a stable matching)
  and the maximum value can be 100, due to the choice of the denominator ($|E| - |M|$).
  Since matchings $M_m$ and $M_p$ are not stable, this  parameter  is  expected to be positive and we 
  consider it as the {\it price} we pay to get positive values for  $\Delta$, $\Delta_1$ and $\Delta_{\RR}$.
\end{itemize}
\end{itemize}

\noindent We now present our results on data-sets generated using different models.
In each case we start with a stable marriage instance (with $|\RR| = |\HH|$)
and gradually increase the capacity.
As before, a column with the legend $\uparrow$ implies that larger values are better for that column. Analogously, a column with the legend $\downarrow$ implies
smaller values are better.

\begin{table}[ht]
\begin{tabular}{|c||c||c|c|c|c||c|c|c|c|}
\hline
\multicolumn{2}{|c|}{$|\RR|=1000, k=5$}&\multicolumn{4}{c||}{$M_p$ vs $M_s$}&\multicolumn{4}{c|}{$M_m$ vs $M_s$} \\
\hline
\hline
  {$|\HH|$} & $S(M_s)$ & $\Delta \uparrow $ & $BP(M_p) \downarrow$ & $\Delta_{1} \uparrow$ & $\Delta_{\RR} \uparrow$ &
                       $\Delta \uparrow $ & $BP(M_m) \downarrow $ & $\Delta_{1} \uparrow $ & $\Delta_{\RR} \uparrow $ \\
\hline
  {$1000$} & 757.90 & {\bf 11.81 } & 4.66 & -3.49 & 5.25 & {\bf 12.79 } & 5.34 & -4.02 & 6.41 \\
  {$100$} & 823.50 & {\bf 12.93 } & 8.57 & -1.64 & 7.41 & {\bf 13.99 } & 9.96 & -2.80 & 8.58 \\
  {$20$} & 870.70 & {\bf 11.65 } & 12.22 & 0.24 & 7.32 & {\bf 12.25 } & 14.47 & -0.36 & 7.47 \\
  {$10$} & 890.00 & {\bf 10.68 } & 16.32 & 0.76 & 2.37 & {\bf 10.80 } & 16.57 & 0.73 & 2.64 \\
\hline
\end{tabular}
\vspace{0.1in}
\caption{Data generated using the \Master model. All values except $S(M_s)$ are percentages.}
\label{tab:SM1-1000}
\vspace{-0.2in}
\end{table}
\noindent {\Master}:
Table~\ref{tab:SM1-1000} shows our results on data-sets generated using the \Master model.
Here, we see that for all data-sets we get at least 10.5\% increase
in the size when comparing $M_s$ vs $M_p$ (column~3) and $M_s$ vs $M_m$ (column~7).
The negative value of $\Delta_1$ (columns~5 and 9) indicates that
we reduce the number of residents matched to their rank-1 hospitals by at most $4.02$\% in our
experiments and we also marginally gain for smaller values of $|\HH|$.
The parameter $\Delta_{\RR}$ is observed to be positive which shows that a majority of residents
prefer $M_p$ over $M_s$ (column~6) and also prefer $M_m$ over $M_s$ (column~10).
Finally, the value of $BP$ (columns~4 and 8) goes on increasing as we reduce the value of $|\HH|$.

\noindent {\Shuffle}:
The results obtained on data-sets generated using the \Shuffle model are
presented in Table~\ref{tab:SM2-1000}.
The size gains are at least 6\% when comparing $M_s$ with $M_p$ (column~3) and
$M_m$ (column~7).
Looking at the values of $\Delta_1$ from column~5 and column~9 we see that the
number of residents matched to their rank-1 hospitals are almost always more,
with up to 18\% getting their rank-1 choices.
We also observe that $\Delta_{\RR}$ is always positive, which implies that
majority of the residents prefer $M_p$ and $M_m$ over $M_s$.

\begin{table}[ht]
\begin{tabular}{|c||c||c|c|c|c||c|c|c|c|}
\hline
\multicolumn{2}{|c|}{$|\RR|=1000, k=5$}&\multicolumn{4}{c||}{$M_p$ vs $M_s$}&\multicolumn{4}{c|}{$M_m$ vs $M_s$} \\
\hline
\hline
  {$|\HH|$} & $S(M_s)$ & $\Delta \uparrow$ & $BP(M_p)\downarrow$ & $\Delta_{1} \uparrow$ & $\Delta_{\RR} \uparrow$ &
                       $\Delta \uparrow$ & $BP(M_m)\downarrow$ & $\Delta_{1} \uparrow$ & $\Delta_{\RR} \uparrow$ \\
\hline
  {$1000$} & 776.80 & {\bf 9.39 } & 2.33 & 0.52 & 4.27 & {\bf 10.20 } & 2.80 & -0.14 & 5.39 \\
  {$100$} & 856.00 & {\bf 8.56 } & 3.55 & 8.72 & 7.80 & {\bf 9.23 } & 4.13 & 9.79 & 9.38 \\
  {$20$} & 900.80 & {\bf 7.10 } & 5.50 & 13.87 & 9.86 & {\bf 7.52 } & 6.01 & 15.55 & 11.37 \\
  {$10$} & 935.40 & {\bf 6.03 } & 16.57 & 17.35 & 5.77 & {\bf 6.15 } & 16.76 & 18.02 & 6.32 \\
\hline
\end{tabular}
\vspace{0.1in}
\caption{Data generated using the \Shuffle model. All values except $S(M_s)$ are percentages.}
\label{tab:SM2-1000}
\vspace{-0.3in}
\end{table}


\noindent {\bf Processed HR couples data-set}: Table~\ref{tab:exp3} shows our results
on publicly available HR with couples data-set~\cite{public}.
As seen from the columns with different $\Delta$ values, 
popular matchings perform favourably on all desired parameters on these data-sets. This
is similar to our observations on data generated using the \Shuffle model. We also investigated
the relation between data generated using \Shuffle model and the data used in this experiment and found
that the data-sets are similar in their characteristics. This confirms
that \Shuffle is a reasonable model. Our results on these data-sets confirm that popular matchings perform favourably on
variants of the \Shuffle model.

\begin{table}[ht]
\begin{tabular}{|c||c||c|c|c|c||c|c|c|c|}
\hline
\multicolumn{2}{|c|}{$|\RR|=100, k=3\ldots5$}&\multicolumn{4}{c||}{$M_p$ vs $M_s$}&\multicolumn{4}{c|}{$M_m$ vs $M_s$} \\
\hline
\hline
  {$|\HH|$} & $S(M_s)$ & $\Delta \uparrow $ & $BP(M_p) \downarrow$ & $\Delta_{1}\uparrow$ & $\Delta_{\RR} \uparrow$ &
                       $\Delta \uparrow $ & $BP(M_m) \downarrow$ & $\Delta_{1}\uparrow$ & $\Delta_{\RR} \uparrow$ \\
\hline
  {$90$} & 84.67 & {\bf 10.10 } & 6.26 & 2.79 & 4.62 & {\bf 11.81 } & 8.55 & 1.20 & 7.16 \\
  {$50$} & 87.19 & {\bf 9.61 } & 8.25 & 4.53 & 4.99 & {\bf 11.01 } & 10.77 & 3.12 & 6.47 \\
  {$20$} & 91.35 & {\bf 7.92 } & 13.57 & 9.30 & 4.28 & {\bf 8.41 } & 14.93 & 8.22 & 3.86 \\
  {$10$} & 93.53 & {\bf 6.61 } & 19.94 & 5.34 & -1.86 & {\bf 6.73 } & 20.43 & 4.99 & -2.00 \\
\hline
\end{tabular}
\vspace{0.1in}
\caption{Processed data-sets from~\cite{public}. All values except $S(M_s)$ are
percentages.}
\label{tab:exp3}
\vspace{-0.2in}
\end{table}


\noindent{\bf Real world data-sets from IIT-M}: 
Table~\ref{tab:real-iitm} shows our results on data-sets obtained from the IIT-M elective allocation (the three rows in the table correspond to
the Aug--Nov~2016, Jan--May~2017 and Aug--Nov~2017 humanities elective allocation data respectively).

\begin{table}[ht]
\setlength\tabcolsep{3.3 pt}
\begin{tabular}{|c|c|c|c||c||c|c|c|c||c|c|c|c|}
\hline
\multicolumn{5}{|c|}{$|\RR|, |\HH|, m$, avg. pref. length}&\multicolumn{4}{c||}{$M_p$ vs $M_s$}&\multicolumn{4}{c|}{$M_m$ vs $M_s$} \\
\hline
\hline
  $|\RR|$ & $|\HH|$ & $|E|$ & $apl$ & $S(M_s)$ & $\Delta \uparrow $ & $BP(M_p) \downarrow$ & $\Delta_{1} \uparrow $ & $\Delta_{\RR} \uparrow $ &
                       $\Delta \uparrow $ & $BP(M_m) \downarrow$ & $\Delta_{1} \uparrow $ & $\Delta_{\RR} \uparrow $ \\
\hline
  $483$ & 18 & 5313 & 11.00 & 481 & {\bf 0.41 } & 1.32 & 0 & -0.20 & {\bf 0.41 } & 1.32 & 0 & -0.20  \\
  $729$ & 16 & 4534 & 6.21 & 675 & {\bf 8.00 } & 31.98 & 7.39 & -5.48 & {\bf 8.00 } & 31.98 & 7.39 & -5.48  \\
  $655$ & 14 & 2689 & 4.10 & 487 & {\bf 18.27 } & 16.42 & 14.97 & 7.17 & {\bf 23.81 } & 29.91 & 24.06 & 5.49  \\
\hline
\end{tabular}
\hspace{0.1in}
 \caption{Real data-sets from IIT-M elective allocation. All values except $S(M_s)$ are percentages.}
 \label{tab:real-iitm}
 \vspace{-0.3in}
\end{table}

For each data-set we list the number of students ($|\RR|$), the number of courses ($|\HH|$), the sum of total
preferences ($m$) and the average preference list over the set of students ($apl$). On an average, each course has a capacity of 50. For each course
a custom ranking criteria is used to rank students who have expressed preference in the course.
As seen in Table~\ref{tab:real-iitm}, for the Jan--May 2017 (row~2) and Aug--Nov 2017 (row~3) data-sets,
popular matchings perform very favourably as compared to stable matching.


\newpage
\bibliography{references}

\newpage

\section{Example Instances}
\label{sec:app_example}
\subsection{Instance where a stable matching is favourable over a maximum cardinality popular matching}
\begin{figure}[ht]
\begin{minipage}{0.45\linewidth}
\begin{center}
\begin{tabular}{cccc}
$\bf{r_1:}$&$h_1$\\
$\bf{r_2:}$&$h_2$ &$h_4$&$h_3$\\
$\bf{r_3:}$&$h_1$& $h_2$\\
$\bf{r_4:}$&$h_4$
\end{tabular}
\end{center}
\end{minipage}
\begin{minipage}{0.5\linewidth}
\begin{center}
\begin{tabular}{ccccc}
$\bf{ [0,1] \hspace{0.1in} h_1:}$ &$r_3$&$r_1$\\
$\bf{ [0,1] \hspace{0.1in} h_2:}$ &$r_3$ & $r_2$\\
$\bf{ [0,1] \hspace{0.1in} h_3:}$ &$r_2$\\
$\bf{ [0,1] \hspace{0.1in} h_4:}$ &$r_4$ & $r_2$\\
\end{tabular}
\end{center}
\end{minipage}
\caption{Example instance in which stable matching $M_s$ matches more residents to their rank-1 hospitals
than the maximum cardinality popular matching $M_p$.}
\label{fig:question1}
\end{figure}
Here, we present an HR instance $G = (\RR \cup \HH, E)$ (in fact a stable marriage instance) 
which shows that (i) the stable matching $M_s$ in $G$ matches more residents to their rank-1 hospitals
than  the maximum cardinality popular matching $M_m$ in $G$ and 
(ii) more number of residents prefer $M_s$ over $M_m$. This example shows that theoretically
there are no guarantees on these parameters for a popular matching.

Let  $\RR = \{r_1, \ldots, r_4\}$ and $\HH = \{h_1, \ldots, h_4\}$ with the preferences of the residents and hospitals as given in
Figure~\ref{fig:question1}. All hospitals have an upper quota of $1$. The instance admits a unique stable matching 
$M_s = \{ (r_2, h_2), (r_3, h_1), (r_4, h_4) \}$ whereas the maximum cardinality popular matching is $M_p = \{ (r_1, h_1), (r_2, h_3), (r_3, h_2), (r_4, h_4) \}$.
The stable matching  $M_s$ matches three residents to their rank-1 hospitals whereas $M_p$ matches exactly one resident to its rank-1 hospital.
Similarly two residents prefer $M_s$ over $M_p$ and exactly one resident prefers $M_p$ over $M_s$.

\subsection{Instance where max. cardinality popular matching is favourable over a stable matching}
Consider another instance where $\RR = \{r_1, \ldots, r_5\}$ and $\HH = \{h_1, \ldots, h_5\}$. All
hospitals have a lower quota of 0 and upper quota of 1. The preferences of the residents and the hospitals
are as given in ~\ref{fig:question2}.
The matching $M_s = \{ (r_1, h_4), (r_3, h_1), (r_4, h_5), (r_5, h_3) \}$ is stable in the instance
whereas the matching $M_p = \{ (r_1, h_4), (r_2, h_5), (r_3, h_1), (r_4, h_3), (r_5, h_2) \}$ is a maximum cardinality popular matching in the instance.
The matching $M_p$ has three residents matched to their rank-1 hospitals as opposed to $M_s$ which
has only two residents matched to their rank-1 hospitals. Furthermore, it is easy to see that more
residents prefer $M_p$ over $M_s$ than the other way.
\begin{figure}[ht]
\begin{minipage}{0.45\linewidth}
\begin{center}
\begin{tabular}{cccc}
$\bf{r_1:}$&$h_5$&$h_4$\\
$\bf{r_2:}$&$h_5$&$h_3$\\
$\bf{r_3:}$&$h_1$\\
$\bf{r_4:}$&$h_3$&$h_5$\\
$\bf{r_5:}$&$h_3$&$h_2$&$h_1$
\end{tabular}
\end{center}
\end{minipage}
\begin{minipage}{0.5\linewidth}
\begin{center}
\begin{tabular}{ccccc}
$\bf{[0,1] \hspace{0.1in}h_1:}$ &$r_3$&$r_5$\\
$\bf{[0,1] \hspace{0.1in}h_2:}$ &$r_5$\\
$\bf{[0,1] \hspace{0.1in}h_3:}$ &$r_5$&$r_2$&$r_4$\\
$\bf{[0,1] \hspace{0.1in}h_4:}$ &$r_1$\\
$\bf{[0,1] \hspace{0.1in}h_5:}$ &$r_4$&$r_1$&$r_2$
\end{tabular}
\end{center}
\label{fig:question2}
\end{minipage}

\caption{Example instance in which max. cardinality popular matching $M_p$ matches more residents to their rank-1 hospitals
than the stable matching $M_s$. More residents prefer $M_p$ over $M_s$ that the other way.}
\label{fig:question2}
\end{figure}


\end{document}